%% file: main.tex
\documentclass[11pt]{article}

\usepackage{amssymb,amsmath,amscd,latexsym,epic,eepic,psfrag,psfig,graphicx}
\usepackage{algorithm}
\usepackage{luca}
\usepackage{fullpage}

\newtheorem{definition}{Definition}
\newtheorem{theorem}{Theorem}
\newtheorem{lemma}{Lemma}

\def\qed{\rule{0.4em}{1.4ex}}

\newcommand{\G}{{\mathcal G}}
\newcommand{\init}{\mathit{in}}
\newcommand{\AP}{\mathsf{AP}}
\newcommand{\lab}{{\mathcal L}}
\newcommand{\ov}{\overline}
\newcommand{\pat}{\omega}

\newcommand{\stra}{\pi}
\newcommand{\Stra}{\Pi}

\newcommand{\obciach}{\upharpoonright}

\def\set#1{\ensuremath{\{#1\}}}

\newcommand{\slopefrac}[2]{\leavevmode\kern.1em
  \raise .5ex\hbox{\the\scriptfont0 #1}\kern-.1em
  /\kern-.15em\lower .25ex\hbox{\the\scriptfont0 #2}}

\begin{document}
\title{The Complexity of Coverage\thanks{This
research was suppored in part by the NSF grants
CCR-0132780 and CNS-0720884.}}

\author{Krishnendu Chatterjee$^\S$ \qquad
  Luca de Alfaro$^{\S}$ \qquad
  Rupak Majumdar$^{\dag}$\\[5pt]
\normalsize
  $\strut^\S$ CE, University of California, Santa Cruz, USA\\
\normalsize
  $\strut^\dag$ CS,
  University of California, Los Angeles, USA\\
\normalsize
  \texttt{$\{$c\_krish,luca\}@soe.ucsc.edu, rupak@cs.ucla.edu}
}

\date{
}

\maketitle

\newif
  \iflong
  \longfalse
\newif
  \ifshort
  \shorttrue

\thispagestyle{empty}

\begin{abstract}
We study the problem of generating a test sequence that achieves maximal 
coverage for a reactive system under test.
We formulate the problem as a repeated game between the tester and the system, 
where the system state space is partitioned according to some coverage 
criterion and the objective of the tester is to maximize the set of partitions 
(or coverage goals) visited during the game.
We show the complexity of the maximal coverage problem for non-deterministic 
systems is PSPACE-complete, but is NP-complete for deterministic systems.
For the special case of non-deterministic systems with a re-initializing 
``reset'' action, which represent running a new test input on a re-initialized 
system, we show that the complexity is again co-NP-complete.
Our proof technique for reset games uses randomized testing strategies that
circumvent the exponentially large memory requirement in the deterministic case.
\end{abstract}

\section{Introduction}

Code coverage is a common metric in software and hardware testing
that measures the degree to which an implementation has been tested
with respect to some criterion.
In its simplest form, 
one starts with a model of the program, 
and a partition of the behaviors of the model
into {\em coverage goals} \cite{AmmannOffutt}. 
A {\em test} is a sequence of inputs that determines a behavior of
the program. 
The aim of testing is to explore as many coverage goals
as possible, ideally as quickly as possible.
In this paper, we give complexity results for several coverage
problems. 
The problems are very basic in nature: they consist in deciding
whether a certain level of coverage can be attained in a given
system. 
It is thus somewhat surprising that the problems have not been
considered previously in the literature. 

Finite-state directed graphs have been used as program
models for test generation of reactive systems
for a long time (see \cite{LeeYannakakis,Brinksma} for surveys).
A coverage goal is a partition of the states of the graph, and
a test is a sequence of labels that determine a path in the graph.
The maximal coverage test generation problem is to hit as 
many partitions as possible using a minimum number of tests.
In the special case the partitions coincide with the states, 
the maximal coverage problem reduces to the Chinese postman problem
for which there are efficient (polynomial time) algorithms
\cite{EdmondsJohnson}. 
In this paper, we show that the maximal coverage problem becomes
NP-complete for graphs with general partitions. 
We also distinguish between {\em system complexity} (the complexity
of the problem in terms of the size of the graph) and the {\em coverage
complexity} (the complexity of the problem in terms of the number of
coverage goals).
Then, the problem is NLOGSPACE in the size of the graph (but
that algorithm uses space polynomial in the number of propositions).

We consider the special case where the graph has a special ``reset'' action
that takes it back to the initial state.
This corresponds in a testing setting to the case where the system can
be re-initialized before running a test.
In this case, the maximal coverage problem remains polynomial, even with
general partitions. 

Directed graphs form a convenient representation for deterministic systems,
in which all the choices are under the control of the tester.
Testing of non-deterministic systems in which certain actions are 
controllable (under the control of the tester) and other actions
are uncontrollable lead to {\em game graphs} \cite{Yannakakis04}.
A game graph is a directed labeled graph where the nodes are partitioned
into tester-nodes and system-nodes, and while the tester can choose the
next input at a tester node,  the system non-deterministically chooses
the next state at a system node.
Then, the test generation problem is to generate a test set that
achieves maximal coverage no matter how the system moves.
For general game graphs, we show the complexity of the maximal coverage
problem is PSPACE-complete.
However, there is an algorithm that runs in time linear in the size
of the game graph but exponential in the number of coverage goals.
Again, the re-initializability assumption reduces the complexity of
coverage: in case there is a re-initialization strategy of the
tester from any system state, the maximal coverage problem for games is
co-NP-complete. 
Dually, we show that the problem of whether it is possible to win a
safety game while visiting fewer than a specified number of partitions
is NP-complete. 

Finally, we consider the coverage problem in bounded time, consisting
in checking whether a specified number of partitions can be visited in
a pre-established number of steps. 
We show that the problem is NP-complete for graphs, and is PSPACE-complete
for game graphs. 

Optimization problems arising out of test generation have been studied
before in the context of both graphs and games 
\cite{AlurSTOC95,LeeY96,Yannakakis04,BlassGNV05}.
However, to the best of our knowledge, the complexities of the coverage
problems studied here have escaped attention so far. 

While we develop our theory for the finite-state, discrete case, we can
derive similar results for more general models, such as those incorporating
incomplete information (the tester can only observe part of the system
state) or timing. 
For timed systems modeled as timed automata, the maximal coverage problem
is PSPACE-complete.
For timed games as well as for (finite state) game graphs with incomplete
information, the maximal coverage problem becomes EXPTIME-complete.

\section{Definitions}

In this section we define \emph{labeled graphs} and \emph{labeled games}, and 
then define the two decision problems of coverage, namely,
\emph{maximal coverage} problem and 
\emph{coverage with bounded time} problem.
We start with definition of graphs and games.

\begin{definition}[Labeled graphs]
A \emph{labeled graph} $\G=((V,E), v_{\init},\AP, \lab)$ consists of the following 
component:
\begin{enumerate}
\item A finite directed graph with vertex set $V$ and edge set $E$;
\item the initial vertex $v_{\init}$;
\item a finite set of atomic propositions $\AP$;
\item a labeling function $\lab$ that assigns to each vertex $v$ the 
set $\lab(v)$ of atomic propositions true at $s$. 
\end{enumerate}
For technical convenience we will assume that for all vertices $v \in V$,
there exists $u \in V$ such that $(v,u)\in E$, i.e., each vertex has at least
one out-going edge.
\end{definition}

\medskip\noindent{\bf Paths in graphs and reachability.} 
Given a labeled graph $\G$, a \emph{path} $\pat$ in $\G$ is a infinite 
sequence of vertices  $\langle v_0,v_1,v_2 \ldots \rangle$ starting
from the initial vertex $v_{\init}$ (i.e., $v_0=v_{\init}$) such that 
for all $i \geq 0$ we have $(v_i,v_{i+1}) \in E$.
A vertex $v_i$ is reachable from $v_{\init}$ if there is a path 
$\pat=\langle v_0,v_1,v_2\ldots \rangle $ in $\G$ and 
$j\geq 0$ such that the vertex $v_j$ in $\pat$ is the vertex $v_i$.

\begin{definition}[Labeled game graphs]
A \emph{labeled game graph} $\G=((V,E), (V_1,V_2), v_{\init},\AP, \lab)$ consists of the 
components of a labeled graph along with a partition of the finite vertex
set $V$ into $(V_1,V_2)$.
The vertices in $V_1$ are player~1 vertices where player~1 chooses
outgoing edges, 
and analogously, the vertices in $V_2$ are player~2 vertices where player~2 chooses
outgoing edges. 
Again for technical convenience we will assume that for all vertices $v \in V$,
there exists $u \in V$ such that $(v,u)\in E$, i.e., each vertex has at least
one out-going edge.
\end{definition}

\medskip\noindent{\bf Plays and strategies in games.} A \emph{play} in a game graph
is a path in the underlying graph of the game.
A strategy for a player in a game is a recipe to specify how to extend the
prefix of a play.
Formally, a strategy $\stra_1$ for player~1 is a function
$\stra_1: V^* \cdot V_1 \to V$ that takes a finite sequence of
vertices $w \cdot v$ ending in a player~1 vertex $v$, where $w\in V^*$ and $v\in V_1$, 
representing the history of the play so far, and specifies the next 
vertex $\stra_1(w\cdot v)$ choosing an out-going edge (i.e., 
$(v,\stra_1(w\cdot v)) \in E$.
A strategy $\stra_2:V^* \cdot V_2 \to V$ is defined analogously.
We denote by $\Stra_1$ and $\Stra_2$ the set of all strategies for 
player~1 and player~2, respectively.
Given strategies $\stra_1 $ and $\stra_2$ for player~1 and player~2, 
there is a unique play (or a path) $\pat(v_{\init},\stra_1,\stra_2)
=\langle v_0,v_1,v_2, \ldots$ such that 
(a)~$v_0=v_{\init}$;
(b)~for all $i\geq 0$, if $v_i \in V_1$, then 
$\stra_1(v_0\cdot v_1 \ldots \cdot v_i)=v_{i+1}$;
and if $v_i \in V_2$, then
$\stra_2(v_0\cdot v_1 \ldots \cdot v_i)=v_{i+1}$.

\medskip\noindent{\bf Controllably recurrent graphs and games.}
Along with general labeled graphs and games, we will also consider
graphs and games that are \emph{controllably recurrent}.
A labeled graph $\G$ is \emph{controllably recurrent} if for every vertex
$v_i$ that is reachable from $v_{\init}$, there is a path starting
from $v_i$ that reaches $v_{\init}$.
A labeled game graph $\G$ is \emph{controllably recurrent} if for every vertex
$v_i$ that is reachable from $v_{\init}$ in the underlying graph, 
there is a strategy $\stra_1$ for player~1 such that against all player~2
strategies $\stra_2$, the path starting from $v_i$ given the
strategies $\stra_1$ and $\stra_2$ reaches $v_{\init}$.
Controllable recurrence models the natural requirement that systems under
test are {\em re-initializable}, that is, from any reachable state of the system,
there is always a way to bring the system back to its initial state no matter
how the system behaves.

\medskip\noindent{\bf The maximal coverage problem.}
The \emph{maximal coverage problem} asks whether at least 
$m$ different propositions can be visited.
We now define the problem formally for graphs and games.
Given a path $\pat=\langle v_0,v_1,v_2,\ldots\rangle$, let 
$\lab(\pat)=\bigcup_{i \geq 0} \lab(v_i)$ be the set of
propositions that appear in $\pat$.
Given a labeled graph $\G$ and $0\leq m \leq |\AP|$, 
the maximal coverage problem asks whether there
is path $\pat$ such that $|\lab(\pat)| \geq m$.
Given a labeled game graph $\G$ and $0\leq m \leq |\AP|$,
the maximal coverage problem asks whether 
player~1 can ensure that at least $m$ propositions 
are visited, i.e.,
whether 
\[
\sup_{\stra_1 \in \Stra_1} \inf_{\stra_2 \in \Stra_2} 
|\lab(\pat(v_{\init},\stra_1,\stra_2))| \geq m.
\]
It may be noted that
$\sup_{\stra_1 \in \Stra_1} \inf_{\stra_2 \in \Stra_2} 
|\lab(\pat(v_{\init},\stra_1,\stra_2))| \geq m$ 
iff there exists a player~1 strategy $\stra_1^*$ such that
for all player~2 strategies $\stra_2^*$ we have 
$|\lab(\pat(v_{\init},\stra_1^*,\stra_2^*))| \geq m$.

The maximal {\em state coverage} problem is the special case
of the maximal coverage problem where $\AP = V$ and 
for each $v\in V$ we have $\lab(v) = \set{v}$. That is,
each state has its own label, and there are $|V|$ singleton partitions.

\medskip\noindent{\bf The coverage with bounded time problem.}
The \emph{coverage with bounded time problem} asks whether 
at least $m$ different propositions can be visited within 
$k$-steps.
We now define the problem formally for graphs and games.
Given a path $\pat=\langle v_0, v_1,v_2,\ldots \rangle$ and $k\geq0$,
we denote by $\pat \obciach k$ the prefix of the path of length $k+1$,
i.e., $\pat \obciach k= \langle v_0,v_1,\ldots,v_k\rangle$.
Given a path $\pat=\langle v_0, v_1,v_2,\ldots \rangle$ and $k\geq0$,
we denote by $\lab(\pat \obciach k) =\bigcup_{0\leq i \leq k} \lab(v_i)$.
Given a labeled graph $\G$ and $0\leq m \leq |\AP|$ and $k\geq 0$,
the coverage with bounded time problem asks whether there
is path $\pat$ such that $|\lab(\pat \obciach k)| \geq m$.
Given a labeled game graph $\G$ and $0\leq m \leq |\AP|$,
the maximal coverage problem asks whether 
player~1 can ensure that at least $m$ propositions 
are visited within $k$-steps, i.e.,
whether 
\[
\sup_{\stra_1 \in \Stra_1} \inf_{\stra_2 \in \Stra_2} 
|\lab(\pat(v_{\init},\stra_1,\stra_2) \obciach k)| \geq m.
\]
It may be noted that
$\sup_{\stra_1 \in \Stra_1} \inf_{\stra_2 \in \Stra_2} 
|\lab(\pat(v_{\init},\stra_1,\stra_2)\obciach k )| \geq m$ 
iff there exists a player~1 strategy $\stra_1^*$ such that
for all player~2 strategies $\stra_2^*$ we have 
$|\lab(\pat(v_{\init},\stra_1^*,\stra_2^*)\obciach k )| \geq m$.

\newcommand{\Sys}{\mathcal{S}}
\newcommand{\Alpha}{\Sigma}

\medskip\noindent{\bf System-tester game.} 
A \emph{system} $\Sys=(Q,\Alpha,q_{\init},\Delta,\AP,\lab)$ consists of 
the following components:
\begin{itemize}
\item A finite set $Q$ of states with the starting state $q_{\init}$.
\item A finite alphabet $\Alpha$ of input letters.
\item A transition relation $\Delta \subseteq Q\times \Sigma \times Q$.
\item A finite set of atomic propositions $\AP$ and a labeling function $\lab$
that assigns to each state $q$ the set of atomic propositions true at $q$.
\end{itemize} 
We consider \emph{total} systems such that for all $q \in Q$ and $\sigma \in
\Alpha$, there exists $q' \in Q$ such that $(q,\sigma,q')\in \Delta$.
A system is \emph{deterministic} if for all $q \in Q$ and $\sigma \in \Alpha$,
there exists exactly one $q'$ such that $(q,\sigma,q')\in \Delta$.
The tester selects an input letter at every stage and the system resolves
the non-determinism in transition to choose the successor state.
The goal of the tester is to visit as many different propositions as possible.
The interaction between the system and the tester can be reduced to a  
labeled game graph $\G=((V,E),(V_1,V_2),v_{\init},\AP,\lab')$ as follows:
\begin{itemize}
\item \emph{Vertices and partition.}
 $V=Q \cup Q\times \Alpha$;  $V_1=Q$ and $V_2=Q\times \Alpha$; and 
 $v_{\init}=q_{\init}$.

\item \emph{Edges.} 
$E=\set{(q,(q,\sigma)) \mid q\in Q, \sigma \in \Alpha} 
\cup \set{((q,\sigma),q') \mid (q,\sigma,q') \in \Delta}$.

\item \emph{Labeling.} 
$\lab'(q)=\lab(q)$ and $\lab'((q,\sigma))=\lab'(q)$.

\end{itemize}
The coverage question for game between tester and system can be 
answered by answering the question in the game graph.
Also observe that if the system is deterministic, then for all player~2 
vertices in the game graph, there is exactly one out-going edge, and hence
the game can be reduced to a labeled graph.
In this paper we will present all the results for the labeled graph
and game model.
All the upper bounds we provide follow also for the game between tester and 
system.
All the lower bounds we present can also be easily adapted to the model 
of the game between system and tester.

\newcommand{\next}{\mathsf{next}}
\newcommand{\nats}{\mathbb{N}}

\section{The Complexity of Maximal Coverage Problems}

In this section we study the complexity of the maximal coverage problem.
In subsection~\ref{subsec:maxexpl:graphs} we study the complexity for 
graphs,
and in subsection~\ref{subsec:maxexpl:games} we study the complexity for 
game graphs.

\subsection{Graphs}
\label{subsec:maxexpl:graphs}
We first show that the maximal coverage problem for labeled
graphs is NP-complete.

\begin{theorem}\label{thrm:maxexpl-graphs}
The maximal coverage problem for labeled graphs is NP-complete.
\end{theorem}
\begin{proof}
The proof consists of two parts. We present them below.
\begin{enumerate}
\item \emph{In NP.} The maximal coverage problem is in NP can be proved
as follows. 
Given a labeled game graph $\G$, let $n=|V|$. 
We show first that if there is a path $\pat$ in $\G$ such that
$|\lab(\pat)| \geq m$, then there is a path $\pat'$ in $\G$ such that
$|\lab(\pat' \obciach m\cdot n) | \geq m$.
If $\pat$ visits at least $m$ propositions, and there is a cycle in $\pat$ 
that does not visit a new proposition that is already visited in the prefix,
then the cycle segment can be removed from $\pat$ and still the resulting
path visits $m$ propositions.
Hence if the answer to the maximal coverage problem is ``Yes'', then there 
is a path $\pat'$ of length at most $m \cdot n$ that is a witness to the 
``Yes'' answer.
Since $m \leq |\AP|$, it follows that the problem is in NP.

\item \emph{NP-hardness.} Now we show that the maximal coverage problem 
is NP-hard, and we present a reduction from the SAT-problem.
Consider a SAT formula $\Phi$, and let $X=\set{x_1,x_2,\ldots,x_n}$ be 
the set of variables and $C_1,C_2,\ldots,C_m$ be the set of clauses.
For a variable $x_j \in X$, let 
\begin{enumerate}
\item $\true(x_j)=\set{\ell \mid x_j \in C_{\ell}}$ be the set of indices 
of the set of clauses $C_{\ell}$ that is satisfied if $x_j$ is set to
be true; and
\item $\false(x_j)=\set{\ell \mid \ov{x}_j \in C_{\ell}}$ be the set of 
indices of the set of clauses $C_{\ell}$ that is satisfied if $x_j$ is set to
be false.
\end{enumerate}
Without loss of generality, we assume that $\true(x_j)$ and $\false(x_j)$ 
are non-empty for all $1\leq j \leq n$ (this is because, for example, 
if $\false(x_j)=\emptyset$, then we can set $x_j$ to be true and 
reduce the problem where the variable $x_j$ is not present).
For a finite set $F \subseteq \nats$ of natural numbers, 
let $\max(F)$ and $\min(F)$ denote the maximum and minimum number of
$F$, respectively,
for an element $f \in F$ that is not the maximal element 
let $\next(f,F)$ denote the next highest element to $f$ that belongs to $F$;
i.e., (a)~$\next(f,F) \in F$; (b)~$f < \next(f,F)$; and
(c)~if $j \in F$ and $f < j$, then $\next(f,F) \leq j$.
We construct a labeled game graph $\G^{\Phi}$ as follows. 
We first present an intuitive description: there are states labeled $x_1,x_2,
\ldots,x_n,x_{n+1}$, and all of them are labeled by a single proposition. 
The state $x_{n+1}$ is an absorbing state (state with a self-loop only), 
and all other $x_i$ state has two successors. The starting is $x_1$.
In every state $x_i$ given the right choice we visit in a line a set of
states that are labeled by clauses that are true if $x_i$ is true; and 
given the left choice we visit in a line a set of
states that are labeled by clauses that are true if $x_i$ is false;
and then we move to state $x_{i+1}$.
We now formally describe every component of the labeled graph $\G^{\Phi}=
\langle (V^{\Phi},E^{\Phi}), v_{\init}^{\Phi}, \AP^{\Phi}, \lab^{\Phi} \rangle $.
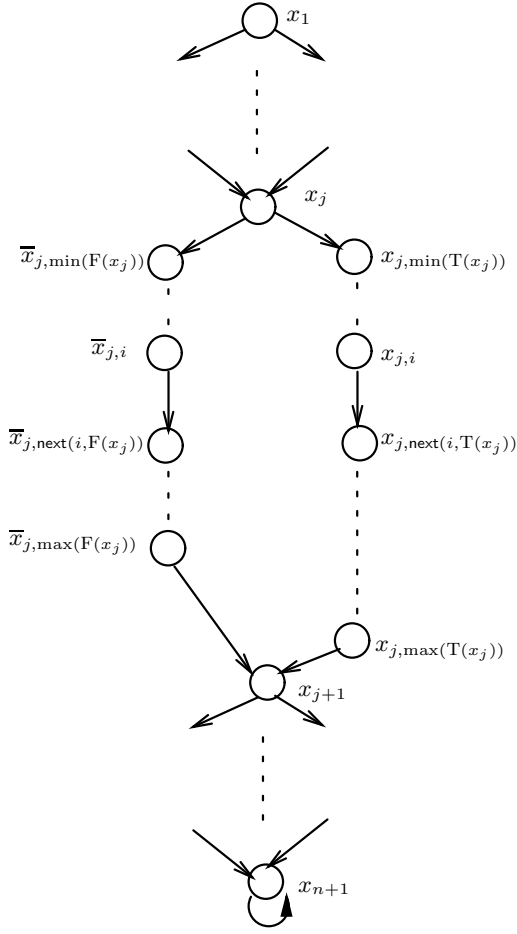
\begin{figure}[t]
   \begin{center}
      \input{np-hard.eepic}
   \end{center}
  \caption{The NP-hardness reduction in picture.}
  \label{fig:np-hardness}
\end{figure}

\begin{enumerate}
\item The set of vertices is 
\[
\begin{array}{rcl}
V^{\Phi} & = & \set{x_i \mid 1 \leq i \leq n+1} \\ 
 & \cup & \set{x_{j,i} \mid 1 \leq j \leq n, i \in \true(x_j)} \
\cup  \ \set{\ov{x}_{j,i} \mid 1 \leq j \leq n, i \in \false(x_j)}.
\end{array}
\]
There is a vertex for every variable, and a vertex $x_{n+1}$.
There is a vertex $x_{j,i}$ iff $C_i \in \true(x_j)$,
and 
there is a vertex $\ov{x}_{j,i}$ iff $C_i \in \false(x_j)$,

\item The set of edges is 
\[
\begin{array}{rcl}
E^{\Phi} & = & \set{(x_{n+1},x_{n+1})} 
\cup \set{(x_{j,\max(\true(x_j))},x_{j+1}), (\ov{x}_{j,\max(\false(x_j))},x_{j+1}) 
\mid 1 \leq j \leq n } \\
& \cup &  \set{(x_j,x_{j,\min(\true(x_j))}), (x_j,\ov{x}_{j,\min(\false(x_j))}) 
\mid 1 \leq j \leq n } \\
& \cup & \set{(x_{j,i},x_{j,\next(i,\true(x_j))} \mid 1\leq j \leq n, i < \max(\true(x_j))} \\
& \cup & \set{(\ov{x}_{j,i},\ov{x}_{j,\false(i,\true(x_j))} \mid 1\leq j \leq n, i < \max(\false(x_j))}.
\end{array}
\]
We now explain the role if each set of edges.
The first edge is the self-loop at $x_{n+1}$.
The second set of edges specifies that from 
$x_{j,\max(\true(x_j))}$ the next vertex is $x_{j+1}$
and similarly, from
$\ov{x}_{j,\max(\false(x_j))}$ the next vertex is again $x_{j+1}$.
The third set of edges specifies that from 
$x_j$ there are two successors that are 
$x_{j,i}$ and $\ov{x}_{j,i'}$ where 
$i=\min(\true(x_j))$ and $i'=\min(\false(x_j))$.
The final sets of edges specifies 
(a) to move in a line from $x_{j,\min(\true(x_j))}$ to visit the clauses that
are satisfied by setting $x_j$ as true, and
(b) to move in a line from $\ov{x}_{j,\false(\true(x_j))}$ to visit the clauses that
are satisfied by setting $x_j$ as false.
Fig~\ref{fig:np-hardness} gives a pictorial view of the reduction.

\item The initial vertex is $v_{\init}^{\Phi}=x_1$.

\item $\AP^{\Phi}=\set{C_1,C_2,\ldots,C_m, X}$, i.e., there is a proposition $C_i$
for each clause $C_i$ and there is a proposition $X$  for all variables;

\item $\lab^{\Phi}(x_j)=X$; i.e., every variable state is labeled by the 
proposition $X$; and we have $\lab^{\Phi}(x_{j,i})=C_i$ and
$\lab^{\Phi}(\ov{x}_{j,i})=C_i$, i.e., each state $x_{j,i}$ and $\ov{x}_{j,i}$
is labeled by the corresponding clause that is indexes.

\end{enumerate}
The number of states in $\G^{\Phi}$ is $O(n \cdot m)$,
and the reduction is polynomial in $\Phi$.
In this graph the maximal number of propositions visited is exactly equal 
to the maximal number of satisfiable clauses plus~1 (since along with the
propositions for clauses the proposition $X$ for all variables is always
visited).
The proof of the above claim is as follows.
Given a path $\pat$ in $G^{\Phi}$ we construct an assignment 
$A$ for the variables as follows: if the choice at a vertex 
$x_j$ is $x_{j,\min(\true(x_j))}$, then we set $x_j$ as true in $A$,
else we set $x_j$ as false.
Hence if a path in $\G^{\Phi}$ visits a set $P \subseteq \AP^{\Phi}$ of
$r$ propositions, 
then the assignment $A$ satisfies $r-1$ clauses (namely, $P \setminus \set{X}$).
Conversely, given an assignment $A$ of the variables, we construct a 
path $\pat^A$ in $\G^{\Phi}$ as follows:
if $x_j$ is true in the assignment $A$, then the path $\pat^A$ chooses
$x_{j,\min(\true(x_j))}$ at $x_j$,
otherwise, it chooses
$\ov{x}_{j,\min(\false(x_j))}$ at $x_j$.
If $A$ satisfies a set $Q$ of $r-1$ clauses, then $\pat^A$ visits $r+1$ propositions
(namely, the set $Q \cup \set{X}$ of propositions).
Hence $\Phi$ is satisfiable iff the answer to the maximal coverage problem
with input $\G^{\Phi}$ and $m+1$ is true.
\end{enumerate}
The desired result follows.
\qed
\end{proof}

\noindent{\bf Hardness of approximation.}
We note that from the proof Theorem~\ref{thrm:maxexpl-graphs} it 
follows that the MAX-SAT problem (i.e., computing the maximal number
of clauses satisfiable for a SAT formula) can be reduced to the problem
of computing the exact number for the  maximal coverage problem.
From hardness of approximation of the MAX-SAT 
problem~\cite{AroraLMSS98}, it follows that the maximal 
coverage problem for labeled graphs is hard to approximate.

\begin{theorem}
The maximal coverage problem for labeled graphs that are 
controllably recurrent can be decided in PTIME.
\end{theorem}
\begin{proof}
To solve the maximal coverage problem for labeled graphs 
that are controllably recurrent, we compute the maximal
strongly connected component $C$ that $v_{\init}$ belongs to.
Since the graph is controllably recurrent, all states
that are reachable from $v_{\init}$ belong to $C$.
Hence the answer to the maximal coverage 
problem is ``Yes'' iff 
$|\bigcup_{v \in C} \lab(v)| \geq m$.
The result follows.
\qed
\end{proof}

\subsection{Game graphs}
\label{subsec:maxexpl:games}

\begin{theorem}\label{thrm:maxexpl-games}
The maximal coverage problem for labeled game graphs is PSPACE-complete.
\end{theorem}
\begin{proof}
The proof consists of two parts. We present them below.
\begin{enumerate}
\item \emph{In PSPACE.} We argue that the maximal coverage problem for labeled game graph 
can be reduced to the coverage in bounded time problem.
The reason is as follows: in a labeled game graph with $n$ 
vertices, if player~1 can visit
$m$ propositions, then player~1 can visit $m$ propositions within 
at most $m\cdot n$ steps; because player~1 can always play a strategy from
the current position that visits a new proposition that
is not visited and never needs to go through a cycle without visiting
a new proposition unless the maximal coverage is achieved.
Hence it follows that the maximal coverage problems for games reduces
to the coverage in bounded time problem.
The PSPACE inclusion will follow from the result of Theorem~\ref{thrm:pspace-bouexpl}
where we show that the coverage in bounded time problem
is in PSPACE.

\item \emph{PSPACE-hardness.} The maximal coverage problem for 
game graphs is PSPACE-complete, even if the underlying graph is 
strongly connected.
The proof is a reduction from QBF (truth of quantified boolean formulas) that
is known to be PSPACE-complete~\cite{Papadimitriou}, 
and it is a modification of the reduction of 
Theorem~\ref{thrm:maxexpl-graphs}.
Consider a QBF formula
\[
\Phi=\exists x_1. \forall x_2. \exists x_3 \ldots \exists x_n. 
C_1 \land C_2 \land \ldots C_m;
\]
defined on the set $X=\set{x_1,x_2,\ldots,x_n}$ of variables,
and $C_1,C_2,\ldots,C_m$ are the clauses of the formula.
We apply the reduction of Theorem~\ref{thrm:maxexpl-graphs} with
the following modification to obtain the labeled game 
graph $\G^{\Phi}$:
the partition $(V_1^{\Phi},V_2^{\Phi})$ of $V^{\Phi}$ is as 
follows.
For a variable $x_j$ if the quantifier before $x_j$ is 
existential, then $x_j \in V_1^{\Phi}$ 
(i.e., for existentially quantified variable, player~1 chooses the
out-going edges denoting whether to set the variable true or false);
and for a variable $x_j$ if the quantifier before $x_j$ is 
universal, then $x_j \in V_2^{\Phi}$ 
(i.e., for universally quantified variable, the opposing player~2 chooses the
out-going edges denoting whether to set the variable true or false).
The state $x_{n+1}$ is a player~2 vertex, and all other vertex has 
an single out-going edges and can be  player~1 state.
Given this game graph we have 
$\Phi$ is true iff player~1 can ensure that all the propositions can be visited
in $\G^\Phi$.
Formally, let $\Stra_1^\Phi$ and $\Stra_2^\Phi$ denote the set 
of all strategies for player~1 and player~2, respectively, in $\G^{\Phi}$.
Then $\Phi$ is true iff 
$\sup_{\stra_1 \in \Stra_1^\Phi}  
\inf_{\stra_2 \in \Stra_2^\Phi} 
|\lab^{\Phi}(\pat(x_1,\stra_1,\stra_2))| \geq m+1$.
Observe that since $x_{n+1}$ is a player~2 state if we add an edge
from $x_{n+1}$ to $x_1$, player~2 will never choose the edge 
$x_{n+1}$ to $x_1$ (since the objective for player~2 is to 
minimize the coverage).
However, adding the edge from $x_{n+1}$ to $x_1$ makes the
underlying graph strongly connected (i.e., the underlying graph of the game
graph becomes controllably recurrent; but player~1 does not have a strategy
to ensure that $x_1$ is reached, so the game is not controllably recurrent).
\end{enumerate}
The desired result follows.
\qed
\end{proof}

\newcommand{\Unif}{\mathsf{Unif}}

\medskip\noindent{\bf Complexity of maximal coverage in controllably recurrent games.}
We will now consider maximal coverage in controllably recurrent games.
Our analysis will use fixing memoryless \emph{randomized} strategy for player~1,
and fixing a memoryless randomized strategy in labeled game graph we get
a labeled Markov decision process (MDP).
A labeled MDP consists of the same components as a labeled game graph, 
and for vertices in $V_1$ (which are randomized vertices in the MDP) 
the successors are chosen uniformly at random (i.e., player~1 does not 
have a proper choice of the successor but chooses all of them uniformly at
random).
Given a labeled game graph $\G=((V,E),(V_1,V_2),v_{\init},\AP,\lab)$ 
we denote by $\Unif(\G)$ the MDP interpretation of $\G$ where player~1 
vertices chooses all successors uniformly at random.
An \emph{end component} in $\Unif(G)$ is a set $U$ of vertices such that
(i)~$U$ is strongly connected and 
(ii)~$U$ is player~1 \emph{closed}, i.e., 
for all $u \in U \cap V_1$, for all $u'$ such that 
$(u,u')\in E$ we have $u' \in U$ (in other words, for all 
player~1 vertices, all the out-going edges are contained in $U$).

\begin{lemma}\label{lemm-contr-mdp}
Let $\G$ be a labeled game graph and let $\Unif(\G)$ be the MDP interpretation
of $\G$.
Then the following assertions hold.
\begin{enumerate}
\item Let $U$ be an end-component in $\Unif(\G)$ with $v_{\init} \in U$.
Then 
$\sup_{\stra_1 \in \Stra_1} \inf_{\stra_2 \in \Stra_2} |\lab(\pat(v_{\init},\stra_1,\stra_2))| 
\leq |\bigcup_{u \in U} \lab(u)|$.

\item There exists an end-component $U \in \Unif(\G)$ with $v_{\init} \in U$
such that $|\bigcup_{u \in U} \lab(u)| \leq \sup_{\stra_1 \in \Stra_1} 
\inf_{\stra_2 \in \Stra_2} |\lab(\pat(v_{\init},\stra_1,\stra_2))| $.
\end{enumerate}
\end{lemma}
\begin{proof} We prove both the claims below.
\begin{enumerate}
\item If $U$ is an end-component in $\Unif(G)$, then 
consider a memoryless strategy $\stra_2^*$ 
for player~2, that for all vertices $u \in U \cap V_2$, 
chooses a successor $u' \in U$ (such a 
successor exists since $U$ is strongly connected).
Since $U$ is player~1 closed (i.e., for all player~1 out-going 
edges from $U$, the end-point is in $U$), 
it follows that for all strategies of player~1 , given
the strategy $\stra_2^*$ for player~2, the vertices
visited in a play is contained in $U$.
The desired result follows.

\item An optimal strategy $\stra_1^*$ for player~1 in $\G$ 
is as follows:
\begin{enumerate}
\item Let $Z_0=\set{v\in V \mid \lab(v)=\lab(v_{\init})}$ and 
$i=0$;
\item At iteration $i$, let $Z_i$ represents the set of propositions already
visited.
At iteration $i$, player~1 plays a strategy to reach 
a state in $V \setminus Z_i$ (if such a strategy exists), and then 
reaches back $v_{\init}$ (a strategy to reach back $v_{\init}$ always
exists since the game is controllably recurrent).

\item If a new proposition $p_i$ is visited at iteration $i$, then 
let $Z_{i+1}=Z_i \cup \set{v \in V \mid \lab(v)=p_i}$.
Goto step (b) for $i+1$ iteration with $Z_{i+1}$.
If no state in $V \setminus Z_i$ can be reached, then stop.
\end{enumerate}
The strategy $\stra_1^*$ is optimal, and let the above iteration
stop with $Z_i=Z^*$.
Let $X=V \setminus Z^*$, and let $X^*$ be the set of vertices
such that player~1 can reach $X$. 
Let $U^*= V \setminus X^*$.
Then $v_{\init} \in U^*$ and player~2 can ensure that from $v_{\init}$
the game can be confined to $U^*$.
Hence the following conditions must hold:
(a)~for all $u \in U^* \cap V_2$, there exists $u' \in U^*$ such 
that $(u,u')\in E$; and
(b)~for all $u \in U^* \cap V_1$, for all $u'\in V$ such that
$(u,u') \in E$ we have $u'\in U^*$.
Consider the sub-graph $G'$ where player~2 restricts itself to edges
only in $U^*$.
A bottom maximal strongly connected component $U \subseteq U^*$ 
in the sub-graph is an end-component in $\Unif(\G)$, and 
we have 
\[
|\bigcup_{u \in U} \lab(u) | 
\leq  |\bigcup_{u \in U^*} \lab(u) | 
\leq  |\bigcup_{u \in Z^*} \lab(u) |.
\] 
It follows that $U$ is a witness end-component to prove the result. 

\end{enumerate}
The desired result follows.
\qed
\end{proof}

\begin{theorem}\label{thrm:maxexpl-contrgames}
The maximal coverage problem for labeled game graphs that are 
controallably recurrent is coNP-complete.
\end{theorem}
\begin{proof}
We prove the following two claims to establish the result.
\begin{enumerate} 

\item \emph{In coNP.} The fact that the problem is in coNP can be proved
using Lemma~\ref{lemm-contr-mdp}.
Given a labeled game graph $\G$, 
if the answer to the maximal coverage problem 
(i.e., whether $\sup_{\stra_1 \in \Stra_1} 
\inf_{\stra_2 \in \Stra_2} |\lab(\pat(v_{\init},\stra_1,\stra_2))|\geq m$) 
is NO, then by Lemma~\ref{lemm-contr-mdp}, there exists 
an end-component $U$ in $\Unif(\G)$ such that 
$|\bigcup_{u \in U} \lab(u)| < m$.
The witness end-component $U$ is a polynomial witness and it can be 
guessed and verified in polynomial time.
The verification that $U$ is the correct witness is as follows:
we check 
(a)~$U$ is strongly connected;
(b)~for all $u \in U \cap V_1$ and for all $u'\in V$ such that 
$(u,u')\in E$ we have $u'\in U$; and
(c)~$|\bigcup_{u\in U} \lab(u) | < m$.
Hence the result follows.

\item \emph{coNP hardness.} 
We prove hardness using a reduction from the complement of 
the \emph{Vertex Cover} problem.
Given a graph $G=(V,E)$, a set $U \subseteq V$ is a \emph{vertex cover} if
for all edges $e=(u_1,u_2)\in E$ we have either $u_1 \in U$  or $u_2 \in U$.
Given a graph $G$ whether there is a vertex cover $U$ of size at most 
$m$ (i.e., $|U| \leq m$) is NP-complete~\cite{GareyJohnson}.
We now present a reduction of the complement of the vertex cover 
problem to the maximal coverage problem in controallably recurrent games.
Given a graph $G=(V,E)$ we construct a labeled game graph $\ov{\G}$ as follows.
Let the set $E$ of edges be enumerated as $\set{e_1,e_2,\ldots,e_{\ell}}$,
i.e., there are $\ell$ edges.
The labeled game graph $\ov{\G}=((\ov{V},\ov{E}),(\ov{V}_1,\ov{V}_2), v_{\init},\AP, \lab)$ 
is as follows.
\begin{enumerate}
\item \emph{Vertex set and partition.} The vertex set $\ov{V}$ is as follows:
\[
\ov{V}=\set{v_{\init}}\cup E \cup \set{e_i^j \mid 1 \leq i \leq \ell, 1 \leq j \leq 2}.
\]
All states in $E$ are player~2 states, and the other states are player~1 states,
i.e., $\ov{V}_2=E$, and $\ov{V}_1=\ov{V}\setminus \ov{V}_2$.

\item \emph{Edges.} The set $\ov{E}$ of edges are as follows:
\[
\begin{array}{rcl}
\ov{E} & = & \set{(v_{\init},e_j) \mid 1 \leq j \leq \ell} 
\ \cup \ \set{(e_i,e_i^j) \mid 1 \leq i \leq \ell, 1 \leq j \leq 2} \\
& & \cup \set{(e_i^j,v_{\init}) \mid 1 \leq i \leq \ell, 1 \leq j \leq 2}.
\end{array}
\]
Intuitively, the edges in the game graph are as follows: from the initial
vertex $v_{\init}$, player~1 can choose any of the edges $e_i \in E$.
For a vertex $e_i$ in $\ov{V}$, player~2 can choose between two vertices
$e_i^1$ and $e_i^2$ (which will eventually represent the two end-points of
the edge $e_i$).
From vertices of the form $e_i^1$ and $e_i^2$, for $1\leq i \leq \ell$, the
next vertex is the initial vertex $v_{\init}$.
It follows that from all vertex the game always comes back to $v_{\init}$ and
hence we have controllably recurrent game.

\item \emph{Propositions and labelling.} $\AP=V \cup\set{\$ \mid \$ \not \in V}$,
i.e., there is a proposition for every vertex in $V$ and a special proposition $\$$.
The vertex $v_{\init}$ and vertices in $E$ are labeled by the  special 
proposition $\$$, i.e., 
$\lab(v_{\init})=\$$; and for all $e_i \in E$ we have $\lab(e_i)=\$$.
For a vertex $e_i^j$, let $e_i=(u_i^1,u_i^2)$, where $u_i^1,u_i^2$ are vertices
in $V$,
then $\lab(e_i^1)=u_i^1$ and $\lab(e_i^2)=u_i^2$.
Note that the above proposition assignment ensures that at every vertex that
represents an edge, player~2 has the choices of vertices that form the end-points
of the edge.

\end{enumerate}

The following case analysis completes the proof.

\begin{itemize}
\item Given a vertex cover $U$, consider a player~2 strategy, that at a 
vertex $e_i \in \ov{V}$, choose a successor $e_i^j$ such that 
$\lab(e_i^j) \in U$.
The strategy for player~2 ensures that player~1 visits only propositions
in $U \cup \set{ \$ }$, i.e., at most $|U| + 1$ propositions.

\item Consider a strategy for player~1 that from $v_{\init}$ visits all 
states $e_1, e_2,\ldots,e_{\ell}$ in order.
Consider any counter-strategy for player~2 and let $U \subseteq V$ be 
the set of propositions other than $\$$ visited.
Since all the edges are chosen, it follows that $U$ is a vertex cover.
Hence if all vertex cover in $G$ is of size at least $m$, then 
player~1 can visit at least $m+1$ propositions. 
\end{itemize}
Hence there is a vertex cover in $G$ of size at most $m$ if and 
only if the answer to the maximal coverage problem in $\G$ with 
$m+1$ is NO.
It follows that the maximal coverage problem in controllably recurrent games
is coNP-hard.
\end{enumerate}
The desired result follows.
\qed
\end{proof}

\noindent{\bf Complexity of minimal safety games.} 
As a corollary of the proof of Theorem~\ref{thrm:maxexpl-contrgames} 
we obtain a complexity result about \emph{minimal safety games}.
Given a labeled game graph $\G$ and $m$, the minimal safety game problem
asks, whether there exists a set $U$ such that a player can 
confine the game in $U$ and $U$ contains at most $m$ propositions.
An easy consequence of the hardness proof of 
Theorem~\ref{thrm:maxexpl-contrgames} is minimal safety games
are NP-hard, and also it is easy to argue that minimal safety games
are in NP.
Hence we obtain that the minimal safety game problem is NP-complete.

\section{The Complexity of Coverage in Bounded Time Problem}
In this section we study the complexity of the coverage in bounded 
time problem.
In subsection~\ref{subsec:bouexpl:graphs} we study the complexity for 
graphs,
and in subsection~\ref{subsec:bouexpl:games} we study the complexity for 
game graphs.

\subsection{Graphs}
\label{subsec:bouexpl:graphs}

\begin{theorem}\label{thrm:bouexpl-graphs}
The coverage in bounded time problem for both labeled graphs and 
controllably recurrent labeled graphs is NP-complete.
\end{theorem}
\begin{proof}
We prove the completeness result in two parts below.
\begin{enumerate}
\item \emph{In NP.} Given a labeled graph with $n$ vertices, if there a path $\pat$ such 
that $|\lab(\pat \obciach k)| \geq m$, then there is path $\pat'$ such 
that $|\lab(\pat' \obciach m\cdot n)| \geq m$. 
The above claim follows since any cycle that does not visit any new proposition
can be omitted. 
Hence a path of length $j=\min(k,m\cdot n)$ can be guessed and it can be
then checked in polynomial time if the path of length $j$ visits at least 
$m$ propositions.

\item \emph{In NP-hard.} We reduce the \emph{Hamiltonian-path (HAM-PATH)}~\cite{GareyJohnson} 
problem to the coverage in bounded time problem for labeled graphs. 
Given a directed graph $G=(V,E)$ and an initial vertex $v$,
we consider the labeled graph $\G$ with the directed graph $G$, with
$v$ as the initial state and $\AP=V$ and $\lab(u)=u$ for all $u\in V$,
i.e., each vertex is labeled with an unique proposition.
The answer to the coverage is bounded time with $k=n$ and $m=n$, 
for $n=|V|$ is ``YES'' iff there is a HAM-PATH in $G$ starting from $v$.
\end{enumerate}
The desired result follows.
\qed
\end{proof}

\smallskip\noindent{\bf Complexity in size of the graph.} 
We now argue that the maximal coverage and the coverage in
bounded time problem on labeled graphs
can be solved in non-deterministic log-space in the size of 
the graph, and polynomial space in the size of the atomic propositions.
Given a labeled graph $\G$, with $n$ vertices, we argued in 
Theorem~\ref{thrm:maxexpl-graphs} that if $m$ propositions 
can be visited, then there is a path of length at most 
$m\cdot n$, that visits $m$ propositions.
The path of length $m\cdot n$, can be visited, storing the 
current vertex, and guessing the next vertex, can checking
the set of propositions already visited.
Hence this can be achieved in non-deterministic log-space in
the size of the graph, and polynomial space in the size of 
the proposition set.
A similar argument holds for the coverage in bounded time problem.
This gives us the following result.

\begin{theorem}\label{thrm:nlogmaxexpl}
Given a labeled graph $\G=((V,E),v_{\init},\AP,\lab)$, 
the maximal coverage problem and the coverage in bounded time problem 
can be decided in NLOGSPACE in $|V|+|E|$, and in PSPACE in $|\AP|$.
\end{theorem}

\subsection{Game graphs}\label{subsec:bouexpl:games}

\begin{theorem}\label{thrm:pspace-bouexpl}
The coverage in bounded time problem for labeled game graphs 
is PSPACE-complete.
\end{theorem}
\begin{proof}
We prove the following two cases to prove the result.
\begin{enumerate}
\item \emph{PSPACE-hardness}. It follows from the proof of 
Theorem~\ref{thrm:maxexpl-games} that the maximal coverage problem
for labeled game graphs reduces to the coverage in bounded time problem
for labeled game graphs. 
Since the maximal coverage problem for labeled game graphs is 
PSPACE-hard (Theorem~\ref{thrm:maxexpl-games}), the result follows.

\item \emph{In PSPACE.}
We say that an {\em exploration game tree\/} for a labeled game graph
is a rooted, labeled tree which represents an unfolding of the graph. 
Every node $\alpha$ of the tree is labeled with a pair $(v, b)$, where 
$v$ is a node of the game graph, and $b \subs \AP$ is the set of
propositions that have been visited in a branch leading from the root
of the tree to $\alpha$. 
The root of the tree is labeled with $(v_{\init}, \lab(v_{\init}))$. 
A tree with label $(v, b)$ has one descendant for each $u$ with 
$(v,u) \in E$; the label of the descendant is $(u, b \union \lab(u))$.

In order to check if $m$ different propositions can
be visited within $k$-steps, the PSPACE algorithm traverses the
game tree in depth first order.
Each branch is explored up to one of the two following conditions is
met: (i)~depth $k$ is reached, or 
(ii)~a node is reached, which has the same label as an ancestor in the
tree. 
The bottom nodes, where conditions (i) or (ii) are met, are thus the
leaves of the tree. 
In the course of the traversal, the algorithm computes in bottom-up
fashion the {\em value\/} of the tree nodes. 
The value of a leaf node labeled $(v, b)$ is $|b|$.
For player-1 nodes, the value is the maximum of the values of the
successors; for player-2 nodes, the value is the minimum of the value
of the successors. 
Thus, the value of a tree node $\alpha$ represents the minimum number of
propositions that player~1 can ensure are visited, in the course of a
play of the game that has followed a path from the root of the tree to
$\alpha$, and that can last at most $k$ steps. 
The algorithm returns Yes if the value at the root is at least $m$,
and no otherwise. 

To obtain the PSPACE bound, notice that if a node with label $(v,b)$
is an ancestor of a node with label $(v',b')$ in the tree, we have 
$b \subs b'$: thus, along a branch, the set of propositions appearing
in the labels increases monotonically. 
Between two increases, there can be at most $|\G|$ nodes, due to the
termination condition (ii). 
Thus, each branch needs to be traversed at most to depth 
$1 + |\G| \cdot (|\AP| + 1)$, and the process requires only polynomial
space. 
\end{enumerate}
The result follows.
\qed
\end{proof}

\begin{theorem}\label{thrm:bouexpl-contrgames}
The coverage in bounded time problem for labeled game graphs that are
controllably recurrent is both NP-hard and coNP-hard, and can be decided in
PSPACE.
\end{theorem}
\begin{proof}
It follows from the (PSPACE-inclusion) argument of 
Theorem~\ref{thrm:maxexpl-games} that 
the maximal coverage problem for labeled game graphs that are 
controllably recurrent can be reduced to the coverage in bounded time 
problem for labeled game graphs that are controllably recurrent.
Hence the coNP-hardness follows from Theorem~\ref{thrm:maxexpl-contrgames}, 
and the NP-hardness follows from hardness in labeled graphs that are
controllably recurrent (Theorem~\ref{thrm:bouexpl-graphs}).
The PSPACE-inclusion follows from the general case of labeled
game graphs (Theorem~\ref{thrm:pspace-bouexpl}).
\qed
\end{proof}

Theorem~\ref{thrm:bouexpl-contrgames} shows that for controllably 
recurrent game graphs, the coverage in bounded time problem is both
NP-hard and coNP-hard, and can be decided in PSPACE. 
A tight complexity bound remains an open problem.

\medskip\noindent{\bf Complexity in the size of the game.}
The maximal coverage problem
can alternately be solved in time linear in the size of the game graph 
and exponential in the number of propositions.
Given a game graph $G = ((V,E), (V_1,V_2), v_{\init}, \AP,\lab)$, construct
the game graph $G' = ((V', E'), (V'_1,V'_2), v'_{\init}, \AP,\lab')$
where $V' = V \times 2^{\AP}$, $((v,b), (v',b'))\in E'$ iff
$(v,v')\in E$ and $b' = b\cup \lab(v')$, 
$V_i = \set{(v,b)\mid v\in V_i}$ for $i\in\set{1,2}$,
$v'_{\init} = (v_{\init}, \lab(v_{\init}))$, and $\lab'(v,b) = \lab'(v)$.
Clearly, the size of the game graph $G'$ is linear in $G$ and
exponential in $\AP$.
Now consider a reachability game on $G'$ with the goal 
$\set{(v,b)\mid v\in V\mbox{ and }|b|\geq m}$.
Player-1 wins this game iff the maximal coverage problem is true for
$G$ and $m$ propositions.
Since a reachability game can be solved in time linear in the game, the
result follows.
A similar construction, where we additionally track the length of 
the game so far, shows that the maximal coverage problem with bounded
time can be solved in time linear in the size of the game graph and
exponential in the number of propositions. 

\begin{theorem}\label{thrm:lintimemaxexpl}
Given a labeled game graph $\G=((V,E),(V_1,V_2),v_{\init},\AP,\lab)$
the maximal coverage and the coverage in bounded time problem 
can be solved in linear-time in $O(|V|+|E|)$ and in exponential time
in $|\AP|$.
\end{theorem}

\section{Extensions}

Somewhat surprisingly, despite the central importance of graph
coverage in system verification, several basic complexity questions have 
remained open.
The basic setting of this paper on graphs and games can be extended in various directions,
enabling the modeling of other system features.
We mention two such directions.

\paragraph{Incomplete Information.}
So far, we have assumed that at each step, the tester has complete information about the
state of the system under test.
In practice, this may not be true, and the tester might be able to observe only a part
of the state.
This leads to graphs and games of {\em imperfect information} \cite{Reif84}.
The maximal coverage and the coverage in bounded time problem 
for games of imperfect information can be solved in EXPTIME.
The algorithm first constructs a perfect-information game graph by
subset construction~\cite{Reif84}, and then run the 
algorithm of Theorem~\ref{thrm:lintimemaxexpl}, that is linear in the size 
game graph and exponential in the number of propositions, on the perfect-information game
graph.  
Thus, the complexity of this algorithm is EXPTIME.
The reachability problem for imperfect-information games is 
already EXPTIME-hard \cite{Reif84}, hence we obtain an optimal
EXPTIME-complete complexity.

\paragraph{Timed Systems.}
Second, while we have studied the problem in the discrete, finite-state setting, similar 
questions can be studied for timed systems modeled as timed automata \cite{timedaut}
or timed game graphs \cite{MalerGames}.
Such problems would arise in the testing of real-time systems.
We omit the standard definitions of timed automata and timed games.
The maximal coverage  problem for timed automata 
(respectively, timed games) takes as input a timed automaton $T$ (respectively, a timed game $T$), 
with the locations labeled by a set $\AP$ of propositions, and
a number $m$, and asks whether $m$ different propositions can be visited.
An algorithm for the maximal coverage problem for timed automata constructs the region
graph of the automaton~\cite{timedaut} 
and runs the algorithm of Theorem~\ref{thrm:nlogmaxexpl} on the labeled region graph.
This gives us a PSPACE algorithm.
Since the reachability problem for timed automata is PSPACE-hard, we obtain 
a PSPACE-complete complexity.
Similar result holds for the coverage in bounded time problem for timed automata.
Similarly, the maximal coverage and coverage in bounded time problem for 
timed games can be solved in exponential time by running the 
algorithm of Theorem~\ref{thrm:lintimemaxexpl} on the region game graph.
This gives an exponential time algorithm.
Again, since game reachability on timed games is EXPTIME-hard \cite{hk99},
we obtain that maximal coverage and coverage in bounded time in 
timed games is EXPTIME-complete.

\end{document}

%% file: np-hard.eepic
\setlength{\unitlength}{0.00041667in}
\begingroup\makeatletter\ifx\SetFigFont\undefined%
\gdef\SetFigFont#1#2#3#4#5{%
  \reset@font\fontsize{#1}{#2pt}%
  \fontfamily{#3}\fontseries{#4}\fontshape{#5}%
  \selectfont}%
\fi\endgroup%
{\renewcommand{\dashlinestretch}{30}
\begin{picture}(5771,11778)(0,-10)
\put(0,6016){\makebox(0,0)[lb]{{\SetFigFont{9}{10.8}{\rmdefault}{\mddefault}{\updefault}$\ov{x}_{j,\next(i,\false(x_j))}$}}}
\thicklines
\put(3187,11529){\ellipse{438}{438}}
\put(3167,9161){\ellipse{438}{438}}
\put(1987,8454){\ellipse{438}{438}}
\put(4387,8529){\ellipse{438}{438}}
\put(4387,7329){\ellipse{438}{438}}
\put(1975,7305){\ellipse{438}{438}}
\put(1987,6129){\ellipse{438}{438}}
\put(4456,6156){\ellipse{438}{438}}
\put(4360,3649){\ellipse{438}{438}}
\put(2020,4821){\ellipse{438}{438}}
\put(3283,3115){\ellipse{438}{438}}
\put(3262,579){\ellipse{438}{438}}
\path(3000,11416)(2175,11041)
\path(2368.660,11194.935)(2175.000,11041.000)(2418.316,11085.691)
\path(3375,11416)(3975,11041)
\path(3739.680,11117.320)(3975.000,11041.000)(3803.280,11219.080)
\path(3000,9016)(2175,8566)
\path(2356.964,8733.598)(2175.000,8566.000)(2414.426,8628.251)
\path(3375,9091)(4200,8641)
\path(3960.574,8703.251)(4200.000,8641.000)(4018.036,8808.598)
\path(2025,7066)(2025,6316)
\path(1965.000,6556.000)(2025.000,6316.000)(2085.000,6556.000)
\path(4425,7066)(4425,6391)
\path(4365.000,6631.000)(4425.000,6391.000)(4485.000,6631.000)
\path(4200,3541)(3450,3241)
\path(3650.551,3385.842)(3450.000,3241.000)(3695.118,3274.425)
\path(2100,4591)(3075,3241)
\path(2885.842,3400.434)(3075.000,3241.000)(2983.123,3470.692)
\path(4050,9916)(3300,9316)
\path(3449.927,9512.779)(3300.000,9316.000)(3524.890,9419.075)
\path(2279,9878)(3029,9278)
\path(2804.110,9381.075)(3029.000,9278.000)(2879.073,9474.779)
\dashline{90.000}(2025,8116)(2025,7666)
\dashline{90.000}(4425,8191)(4425,7666)
\dashline{90.000}(4425,5791)(4425,3991)
\dashline{90.000}(2025,5791)(2025,5191)
\dashline{90.000}(3225,2416)(3225,1366)
\path(3158,2924)(2333,2549)
\path(2526.660,2702.935)(2333.000,2549.000)(2576.316,2593.691)
\path(3379,2948)(3979,2573)
\path(3743.680,2649.320)(3979.000,2573.000)(3807.280,2751.080)
\path(2339,1235)(3089,635)
\path(2864.110,738.075)(3089.000,635.000)(2939.073,831.779)
\path(4043,1375)(3293,775)
\path(3442.927,971.779)(3293.000,775.000)(3517.890,878.075)
\dashline{90.000}(3150,10891)(3150,9841)
\put(3525,11491){\makebox(0,0)[lb]{{\SetFigFont{9}{10.8}{\rmdefault}{\mddefault}{\updefault}$x_1$}}}
\put(3750,9166){\makebox(0,0)[lb]{{\SetFigFont{9}{10.8}{\rmdefault}{\mddefault}{\updefault}$x_j$}}}
\put(3675,2866){\makebox(0,0)[lb]{{\SetFigFont{9}{10.8}{\rmdefault}{\mddefault}{\updefault}$x_{j+1}$}}}
\put(3675,391){\makebox(0,0)[lb]{{\SetFigFont{9}{10.8}{\rmdefault}{\mddefault}{\updefault}$x_{n+1}$}}}
\put(1050,7216){\makebox(0,0)[lb]{{\SetFigFont{9}{10.8}{\rmdefault}{\mddefault}{\updefault}$\ov{x}_{j,i}$}}}
\put(150,8341){\makebox(0,0)[lb]{{\SetFigFont{9}{10.8}{\rmdefault}{\mddefault}{\updefault}$\ov{x}_{j,\min(\false(x_j))}$}}}
\put(0,4741){\makebox(0,0)[lb]{{\SetFigFont{9}{10.8}{\rmdefault}{\mddefault}{\updefault}$\ov{x}_{j,\max(\false(x_j))}$}}}
\put(4725,8341){\makebox(0,0)[lb]{{\SetFigFont{9}{10.8}{\rmdefault}{\mddefault}{\updefault}$x_{j,\min(\true(x_j))}$}}}
\put(4725,7141){\makebox(0,0)[lb]{{\SetFigFont{9}{10.8}{\rmdefault}{\mddefault}{\updefault}$x_{j,i}$}}}
\put(4725,6016){\makebox(0,0)[lb]{{\SetFigFont{9}{10.8}{\rmdefault}{\mddefault}{\updefault}$x_{j,\next(i,\true(x_j))}$}}}
\put(4650,3391){\makebox(0,0)[lb]{{\SetFigFont{9}{10.8}{\rmdefault}{\mddefault}{\updefault}$x_{j,\max(\true(x_j))}$}}}
\put(3300.000,271.000){\arc{510.000}{5.7932}{9.9147}}
\blacken\path(3583.031,150.516)(3525.000,391.000)(3463.035,151.500)(3583.031,150.516)
\end{picture}
}

%% file: main.bbl
\begin{thebibliography}{10}

\bibitem{AlurSTOC95}
R.~Alur, C.~Courcoubetis, and M.~Yannakakis.
\newblock Distinguishing tests for nondeterministic and probabilistic machines.
\newblock In {\em Proc. 27th ACM Symp. Theory of Comp.}, 1995.

\bibitem{timedaut}
R.~Alur and D.~Dill.
\newblock The theory of timed automata.
\newblock In {\em Real-Time: Theory in Practice}, volume 600 of {\em Lect.
  Notes in Comp. Sci.}, pages 45--73. Springer-Verlag, 1991.

\bibitem{AmmannOffutt}
P.~Ammann and J.~Offutt.
\newblock {\em Introduction to software testing}.
\newblock Cambridge University Press, 2008.

\bibitem{AroraLMSS98}
S.~Arora, C.~Lund, R.~Motwani, M.~Sudan, and M.~Szegedy.
\newblock Proof verification and the hardness of approximation problems.
\newblock {\em J. ACM}, 45(3):501--555, 1998.

\bibitem{BlassGNV05}
A.~Blass, Y.~Gurevich, L.~Nachmanson, and M.~Veanes.
\newblock Play to test.
\newblock In {\em FATES}, volume 3997 of {\em Lecture Notes in Computer
  Science}, pages 32--46. Springer, 2005.

\bibitem{Brinksma}
E.~Brinksma and J.~Tretmans.
\newblock Testing transition systems: An annotated bibliography.
\newblock In {\em MOVEP 00: Modeling and Verification of Parallel Processes},
  volume 2067 of {\em Lecture Notes in Computer Science}, pages 187--195.
  Springer, 2000.

\bibitem{EdmondsJohnson}
J.~Edmonds and E.~L. Johnson.
\newblock Matching, {E}uler tours and the {C}hinese postman.
\newblock {\em Math.\ Prog.}, 5:88--124, 1973.

\bibitem{GareyJohnson}
M.R. Garey and D.S. Johnson.
\newblock {\em Computers and Intractability: {A} Guide to the Theory of
  NP-Completeness}.
\newblock Freeman and Co., 1979.

\bibitem{hk99}
T.~Henzinger and P.~Kopke.
\newblock Discrete-time control for rectangular hybrid automata.
\newblock {\em Theoretical Computer Science}, 221:369--392, 1999.

\bibitem{LeeY96}
D.~Lee and M.~Yannakakis.
\newblock Optimization problems from feature testing of communication
  protocols.
\newblock In {\em ICNP: International Conference on Network Protocols}, pages
  66--75. IEEE Computer Society, 1996.

\bibitem{MalerGames}
O.~Maler, A.~Pnueli, and J.~Sifakis.
\newblock On the synthesis of discrete controllers for timed systems.
\newblock In {\em Proc. of 12th Annual Symp. on Theor. Asp. of Comp. Sci.},
  volume 900 of {\em Lect. Notes in Comp. Sci.} Springer-Verlag, 1995.

\bibitem{Papadimitriou}
C.H. Papadimitriou.
\newblock {\em Computational Complexity}.
\newblock Addison-Wesley, 1993.

\bibitem{Reif84}
J.H. Reif.
\newblock The compexity of two-player games of incomplete information.
\newblock {\em Journal of Computer and System Sciences}, 29:274--301, 1984.

\bibitem{Yannakakis04}
M.~Yannakakis.
\newblock Testing, optimizaton, and games.
\newblock In {\em LICS}, pages 78--88. IEEE Computer Society, 2004.

\bibitem{LeeYannakakis}
M.~Yannakakis and D.~Lee.
\newblock Testing for finite state systems.
\newblock In {\em CSL 98: Computer Science Logic}, volume 1584 of {\em Lecture
  Notes in Computer Science}, pages 29--44. Springer, 1998.

\end{thebibliography}
